\documentclass[12pt]{amsart}

\usepackage{amsmath, amsthm, latexsym, amssymb, amsfonts, epsf, xcolor, hyperref}

\numberwithin{equation}{section}
\usepackage{bibentry}
\usepackage{amsthm}
\usepackage{enumerate}
\newtheorem{theorem}{Theorem}%
\newtheorem{proposition}[theorem]{Proposition}
\newtheorem{corollary}[theorem]{Corollary}
\newtheorem{example}[theorem]{Example}%

\newtheorem{fact}[theorem]{Fact}
\newtheorem{definition}[theorem]{Definition}

\newcommand{\F}{\mathbb{F}}
\newcommand{\N}{\mathbb{N}}
\newcommand{\supp}{\textnormal{supp}}
\newcommand{\ev}{\textnormal{ev}}
\newcommand{\rmv}[1]{}
\newcommand{\glm}[1]{{\color{magenta}    #1}}

\begin{document}

\title[Fractional decoding of algebraic geometry codes over extension fields]{Fractional decoding of algebraic geometry codes over extension fields
}

\author[E. Camps-Moreno]{Eduardo Camps-Moreno}
\address[Eduardo Camps-Moreno]{Department of Mathematics\\ Virginia Tech\\ Blacksburg, VA USA}
\email{e.camps@vt.edu}

\author[G. L. Matthews]{Gretchen L. Matthews}
\address[Gretchen L. Matthews]{Department of Mathematics\\ Virginia Tech\\ Blacksburg, VA USA}
\email{gmatthews@vt.edu}

\author[W. Santos]{Welington Santos}
\address[Welington Santos]{Department of Mathematics, Statistics and Computer Science University of Wisconsin-Stout}
\email{santosw@uwstout.edu}

\begin{abstract} 
In this paper, we study algebraic geometry codes from curves over $\mathbb{F}_{q^\ell}$ through their virtual projections which are algebraic geometric codes over $\mathbb{F}_q$. We 
use the virtual projections to 
provide fractional decoding algorithms for the codes over $\mathbb{F}_{q^\ell}$. Fractional decoding seeks to perform error correction using a smaller fraction of $\F_q$-symbols than a typical decoding algorithm.
In one instance, the bound on the number of correctable errors differs from the usual lower bound by the degree of a pole divisor of an annihilator function. In another, we view the virtual projections as interleaved codes to, with high probability, correct more errors than anticipated. 
\end{abstract}



\maketitle

\section{Introduction} \label{s:intro}
The fractional decoding problem is motivated by  distributed systems in which there are limitations on the disk operation and on the amount of information transmitted for the purpose of decoding. Sometimes thought of as error correction with partial information, fractional decoding considers codes defined over an extension field and algorithms for error correction that use fewer symbols from the base field than is typical, thus operating using a restricted amount of information in the decoding process. Tamo, Ye, and Barg \cite{BarM} introduced the concept, inspired by regenerating codes and codes for erasure recovery which focus on the repair bandwidth which measures the proportion of a received word used to recovery an erasure (see, for instance, \cite{Dimakis, GW, YeBarg}). 
	
 In \cite{Fractional_ISIT2019}, Santos provided a connection between fractional decoding of Reed-Solomon codes, which can be considered as codes from
 the projective line, a curve of genus $0$, and collaborative decoding of interleaved Reed-Solomon codes. Later in \cite{Fractional_ISIT2021,Fractional_rHerm}, the authors present the first fractional decoding algorithm for codes from the Hermitian curve. The $r$-Hermitian codes are constant field extensions of Hermitian codes with a bound on the $y$-degree of the rational functions in the Riemann-Roch space used to define codewords.

In this paper, we present a fractional decoding approach for a large family of algebraic geometry codes.  The codes may be considered as evaluation codes over the field $\F_{q^l}$ of cardinality $q^l$ with evaluation points whose coordinates that lie in a subfield. Taking evaluation points to have coordinates in a proper subfield (such as $\F_q$ where $l>1$) to define a family of codes for particular purpose is an idea utilized by Guruswami and Xing \cite{Guruswam} and  Gao, Yue, Huang, and Zhang \cite{GYHZ} among others (including \cite{mvG}). 
To be more precise, we focus on algebraic geometry codes defined as follows. Consider $n$ distinct affine points $ P_1, \dots, P_n$  on a smooth, projective curve $\mathcal X$ with coordinates in $\F_q$ and a divisor $G$ of degree less than $n$ on $X$ whose support consists of only $\F_q$-rational points and contains none of the points $P_1, \dots, P_n$. We are interested in the $q^l$-ary  algebraic geometry code $C_l(G, D):=\ev(\mathcal L_l(G)) \subseteq \F_{q^l}^n$ where
$$
\begin{array}{lccc}
\ev: & \F_{q^l}(\mathcal X) & \rightarrow & \F_{q^l}^n \\
& f & \mapsto & (f(P_1), \dots, f(P_n)) 
\end{array}
$$
and $D=P_1+\dots+P_n$.
Noting that $C_l(G, D)$ has length $n$, dimension $\ell(G)$, which is the dimension of the Riemann-Roch space $\mathcal L(G)$, and minimum distance $d$ at least $n- \deg G$, we know that up to $e$  errors may be corrected in a received word $w \in \F_{q^l}^n$, where $e \leq \frac{d-1}{2}$. Hence, viewing $\F_{q^l}$ as a degree $l$ extension of $\F_q$, we may see this as correcting $e$ errors using $ln$ symbols of $\F_q$. 

In this work, we use partitions of the set of evaluation points $P_1, \dots, P_n$ to demonstrate that up to $\frac{d'-1}{2}$ errors may be corrected using only $mn$ symbols of the base field $\F_q$, where $m<l$ is number of parts of a partition and $d'$ is specified by pole divisors of certain rational functions on $\mathcal X$. This allows for error correction using algebraic geometry codes in applications where limiting network traffic is beneficial, such as distributed storage.
We emphasize two key advantages to this approach. First, while quite general, these techniques apply particularly nicely to codes constructed from Kummer extensions or Castle curves and a constant field extension of the Riemann-Roch space of a divisor on the curve. Second, the perspective may be combined with advances in traditional decoding algorithms for algebraic geometry codes, providing immediate adapation to the distributed context considered in this work. 

This paper is organized as follows. Section \ref{S:virtual_projection} presents the main tools to be used in fractional decoding of codes from curves. The codes we consider are described in Subection \ref{S:constant_ext} where we study their virtual projections. These virtual projections may be conveniently captured as an interleaved code, as seen in Subsection \ref{S:interleaved_vp}. Particular instances for important families of codes are given in Subsection \ref{S:vp_Kummer}. Fractional decoding algorithms are presented in Section \ref{S:decoding}. The first, in Subsection \ref{s:directly}, makes use of the virtual projections directly. The second is via interleaved codes found in Subsections \ref{S:interleaved_vp} and \ref{s:decoding_interleaved}. These frameworks can make use of advances in decoding of the component or interleaved codes.

\section{Virtual projections of algebraic geometry codes} \label{S:virtual_projection}

In this section, we will consider convenient ways to represent algebraic geometry codes defined over extention fields. In \cite{Guruswam}, Guruswami and Xing considered algebraic-geometric codes over a field $\F_{q^l}$ whose evaluation points belong to a subfield $\F_q$ and provided a list decoding algorithm for those codes. Here, we consider properties of a similar family of codes, in preparation for  demonstrating that  fractional decoding can be applied to them. 

\subsection{Notation} \label{s:notation}
We use the standard notation from coding theory and the algebraic curves over finite fields; see \cite{Stichtenoth} for instance. Given a genus $g$ curve $\mathcal X$ over a finite field $\F$, the associated field of rational functions is denoted $\F(\mathcal X)$. For a rational function $f \in \F(\mathcal X)$ and $\F$-rational point $P$ on $\mathcal X$, the valuation of $f$ at $P$ is denoted $v_P(f)$. The divisor of a nonzero rational function  $f$ is $(f)=(f)_0 - (f)_{\infty}$ where $(f)_0=\sum_{{\footnotesize{ \{
P \in \mathcal X (\overline{\F}) : v_P(f)>0 \} }}}
v_P(f) P$ and $(f)_{\infty}=\sum_{{\footnotesize{
\{
P \in \mathcal X (\overline{\F}): v_P(f)<0 \}}}}
v_P(f) P$ are the zero and pole divisors of $f$. For a divisor  $D$ on $\mathcal X$, the Riemann-Roch space of $G$ is vector space of functions 
$\mathcal L(G):= \left\{ f \in \F(\mathcal X) : (f) + G \geq 0 \right\}$ whose dimension over $\F$ is denoted $\ell(G)$.  If $\deg G \geq 2g-1$, then $\ell(G)=\deg G + 1 -g$, and $\Omega(G)$ is the vector space of differentials with divisors at least $G$, together with the zero differential. We use $\N$ to denote the set of nonnegative integers. Given a positive integer $n$, $[n]:=\left\{ 1, \dots, n\right\}$. For integers $i, j$, 
$\delta_{i,j}=1$ if and only if $i=j$ and $\delta_{i,j}=0$ otherwise. The set of all $m \times n$ matrices over $\F$ is denoted $\F^{m \times n}$. Given a matrix $A=(a_{ij}) \in \F^{m \times n}$, we let $D_{Row_iA}$ denote the diagonal matrix whose diagonal is given by the $i^{th}$ row of $A$ and $A\mid_{I} \in \F^{m \times \mid I \mid}$ denote the submatrix of $A$ whose columns are the columns of $A$ indexed by $I \subseteq [n]$.

An $[n,k,d]$ code $C$ over a finite field $\F$ is a $k$-dimensional $\F$-subspace of $\F^n$ in which the minimum weight of a nonzero vector is $d$. Here, $n$ is the length of $C$, $k$ is the dimension of $C$, $d$ is the minimum distance of $C$, and the weight of a word $w \in \F^n$, denoted $wt(w)$ is its number of nonzero coordinates. Such a code can be described by a generator matrix, which is a matrix $G \in \F^{k \times n}$ whose rows form a basis for $C$. We say that a subset of $k$ coordinate positions $I \subseteq [n]$ is an information set if and only if the $k \times k$ submatrix of $G$ formed by columns indexed by $I$ is nonsingular.

We will consider algebraic geometry codes $C_l(G, D)$ as described in Section \ref{s:intro}. The points in the support of $D$ are referred to as the evaluation points of the code. We typically take them to be affine points on the curve, in which case an evaluation point may be expressed as $(a_1, \dots, a_M)$ where $\mathcal X$ is given by equations involving variables $x_1, \dots, x_M$.
 In the case where $D$ is the sum of all $\F_q$-rational points not in the support of $G$ or is clear from the context, we write $C_l(G)$ rather than $C_l(G, D)$. If the support of $G$ consists of a single point, the codes $C_l(G, D)$ and $C(G, D)$ are called one-point codes. We say that a collection of points $\left\{ P_{i_1}, \dots, P_{i_k} \right\} \subseteq \left\{ P_1, \dots, P_n \right\}$ is an information set to mean that the positions $\left\{ {i_1}, \dots, {i_k} \right\}$  corresponding to these points is an information set.

\subsection{Codes from constant field extensions} \label{S:constant_ext}

Let $\left\{ \zeta_{1},\ldots,\zeta_{l} \right\}$ be a basis of $F:=\F_{q^l}$ over $B:=\F_q$, and let $\left\{ \nu_{1},\ldots,\nu_{l} \right\}$ be its dual basis, meaning $tr(\zeta_{s}\nu_{j})=\delta_{s,j}$ for all $s, j \in [l]$. Then for all $\beta \in F$,
\[\beta=\sum_{s=1}^{l}tr(\zeta_{s}\beta)\nu_{s}.\]
In other words, any element $\beta$ in $F$ can be calculated from its $l$ projections $tr(\zeta_{s}\beta)$, $s \in [l]$, on $B$.

Consider a smooth projective curve $\mathcal X$ over a field $\F_{q}$ and a divisor $G$ whose support contains only $\F_q$-rational points. 
At times, we may wish consider  the Riemann-Roch space of $G$ on $\mathcal X$ considered over the extension  $\F_{q^l}$, denoted $\mathcal{L}_{l}(G)$. It is well-known that
$$
\mathcal{L}_{l}(G)=\mathcal{L}(G)\otimes\mathbb{F}_{q^{l}}.
$$

\begin{definition}
	Keep the notation above and assume that $\mathcal B=\left\{ h_1, \dots, h_k \right\}$ is a basis for
 $\mathcal L (G)$. 
 For $s \in [l]$,  the $s$-projection of the function $f(x_1,\ldots,x_M)=\sum_{i=1}^k a_i h_i(x_1,\ldots,x_M)\in\mathcal{L}_{l}(G)$ to $\mathcal{L}(G)$ with respect to $\mathcal B$ is defined to be
$$f_s(x_1,\ldots,x_M)=\sum_{i=1}^{k}tr(\zeta_{s}a_{i})h_i(x_1,\ldots,x_M).$$
Typically, a basis for $\mathcal L (G)$ is fixed, and we refer to $f_s$ as the $s$-projection of $f$ to $\mathcal L(G)$, omitting mention of the particular choice of basis.
\end{definition}

Note that $f \in  \mathcal{L}_{l}(G) \subseteq \F_{q^l}(\mathcal X)$ while $f_s 
\in \mathcal{L}(G) \subseteq \F_{q}(\mathcal X)$. Furthermore,  $f(x_1,\ldots,x_M)$ is fully determined by $\left\lbrace f_{s}(x_1,\ldots,x_M): s \in [l]\right\rbrace$, since 
$$ 
\begin{array}{lll}
f(x_1,\ldots,x_M) &= \sum_{i=1}^k a_i h_i(x_1,\ldots,x_M) \\ \ \\ &= \sum_{i=1}^k \left[ \sum_{s=1}^l tr(\zeta_{s}a_{i})\nu_s  \right] h_i(x_1,\ldots,x_M) \\ \ \\
&= \sum_{s=1}^l \left[ \sum_{i=1}^k tr(\zeta_{s}a_{i}) h_i(x_1,\ldots,x_M) \right]\nu_s \\ \ \\ &= \sum_{s=1}^l f_s(x_1,\ldots,x_M)\nu_s. \end{array}$$
This observation will be useful in the fractional decoding of certain algebraic geometry codes. 
It is also straightforward to verify that the $s$-projection of a $\F_q$-linear combination of functions behaves nicely, as mentioned in the next observation. 

\begin{fact} \label{f:s_prop}
For all $f, h \in \mathcal L_l(G)$ and $a,b \in \F_{q}$,
    $$
(af+bh)_s=af_s+bh_s.
$$
\end{fact}

To set the stage for fractional decoding, fix a partition $A_1 \dot{\cup} \cdots \dot{\cup} A_m \subseteq  \left\{ P_1, \dots, P_n\right\}$ where $m<l$. Consider an annihilator polynomial $p_t(x_1, \dots, x_M) \in \F_{q}[x_1, \dots, x_M]$ of the set $A_{t}$, $t \in [m]$, meaning $p(x_1, \dots, x_M)=0$ for all $(x_1, \dots, x_M) \in A_t$. 	

\begin{definition}\label{Tr2}
	Given  $f\in\mathcal{L}_{l}(G)$, $A_{1} \dot{\cup} \cdots \dot{\cup} A_{m}\subseteq \left\{ P_1, \dots, P_n \right\}$, and $t\in [m]$, define the function
 $$
\begin{array}{ll}
T_{t}(f)(x_1, \dots, x_M)=&f_{l-m+t}(x_1, \dots, x_M)(p_{t}(x_1, \dots, x_M))^{l-m} \\ \ \\ &+ \sum_{s=1}^{l-m}f_{s}(x_1, \dots, x_M)(p_{t}(x_1, \dots, x_M))^{s-1}\end{array}$$
and the $t$-virtual projection of $\mathcal{C}_{l}(G,D)$  to be
	$$\mathcal{VP}_{t}(G,D)=\left\lbrace (T_{t}(f)(P_1),\ldots,T_{t}(f)(P_n)):f\in\mathcal{L}_{l}(G)\right\rbrace.$$
If $D$ is the sum of all $\F_q$-rational points not in the support of $G$ or is clear from the context, we write $\mathcal{VP}_{t}(G)$ rather than $\mathcal{VP}_{t}(G,D)$.
\end{definition}

The next result demonstrates how the virtual projection of an algebraic geometry code may be seen as an algebraic geometry code itself. 

\begin{proposition} \label{P:subcode}
Given divisors $G$ and $D$ whose supports consist only of $\F_q$-rational points,
    the $t$-virtual projection $\mathcal{VP}_{t}(G,D)$ of $\mathcal{C}_{l}(G,D)$ is a subcode of $\F_q^n$. In particular, 
    $$
\mathcal{VP}_{t}(G,D) = \ev \left( T_t \left( \mathcal L_l(G) \right) \right) \subseteq 
C \left(D, G - (l-m) \left(p_t\right)_{\infty} \right).
$$    
\end{proposition}

\begin{proof}
 According to Fact \ref{f:s_prop}, for all $f, h \in \mathcal L_l(G)$ and $a,b \in \F_{q}$, $$T_t \circ af + T_t \circ bh = T_t \circ \left( af+bh \right)$$ 
and 
$$
a \cdot \ev \left( T_t(f) \right)+b \cdot \ev \left( T_t(h) \right)= \ev \left( T_t(af+bh) \right).
$$ Now the result follows from the fact that 
$$\mathcal{VP}_{t}(G) = \ev(T_t(\mathcal L_l(G))) \subseteq
\F_q^n.$$
Indeed, note that for any function $f \in \mathcal L_l(G)$ and $t \in [m]$,
\begin{equation} \label{E:vP}
\begin{array}{lcl}
v_P(T_t(f))  &\geq & \min \left\{ 
v_P \left( f_i p_t^i \right), 
  v_P \left( f_{l-m+t} p_t^{l-m}   \right):  
i \in [l-m]  
\right\} \\ \ \\
&\geq&    \min \left\{ \begin{array}{l}
 v_P \left( f_i\right) + i v_P\left(p_t \right), \\
  v_P \left( f_{l-m+t}\right) + (l-m)  v_P \left(p_t \right)\end{array}:  
i \in [l-m] 
\right\}
\\ \ \\
&\geq & -v_P \left(G \right)+   \min \left\{    i v_P\left(p_t \right)  : i \in [l-m] 
\right\}
\\ \ \\
&\geq & \begin{cases}
-v_P(G) & \textnormal{if } v_P\left(p_t \right) \geq 0 \\
-v_P(G)+(l-m)v_P\left(p_t \right) & \textnormal{if } v_P\left(p_t \right) < 0.
\end{cases}
\end{array}
\end{equation}
Hence, $$T_t(f) \in \mathcal L\left( G - (l-m) \sum_{P \in \supp (p_t)_{\infty}} v_P\left(p_t \right) P  \right).$$
\qed 
\end{proof}

Proposition \ref{P:subcode} demonstrates the influence of the partition and the annihilator functions on the virtual projection.
More information about the annihilator functions will allow us to better bound $v_P(T_t(f))$, resulting in a better understanding of the space that contains $\left\{T_t(f): f \in \mathcal L_l(G) \right\}$ for each $t \in [m]$.
 We will see that
a key idea to achieving fractional decoding of $C_l(G)$ is decoding $\mathcal{VP}_{t}(G)$, $t \in [m]$.

\begin{theorem}\label{T:recover1}
 Suppose $I \subseteq A_{1} \dot{\cup} \cdots \dot{\cup} A_{m}\subseteq \left\{ P_1, \dots, P_n \right\}$,
  for some $m \in \N$ and some information set $I$ for the code $C_l(G)$ with generator matrix $A=\left( h_i(P_j)\right)$. Then the function
 $f \in \mathcal{L}_{l}(G)$ depends only on $\left\lbrace T_{t}(f)(P): t \in [m], P \in I \right\rbrace$ and $A\mid_I$.  
\end{theorem}

\begin{proof} Given $A_{1} \dot{\cup} \cdots \dot{\cup} A_{m}\subseteq \left\{ P_1, \dots, P_n \right\}$ and  $\left\lbrace T_{t}(f)(x_1,\ldots,x_M): t \in [m]\right\rbrace$, we aim to determine $f$. 
For $i\in [l-m]$ and $t \in [m]$, let 
$$ T_{t}^{(i)}(f)(x_1, \dots, x_M)=
 \frac{T_{t}(f)(x_1, \dots, x_M)-
\sum_{s=1}^{i-1}f_{s}(x_1, \dots, x_M)p_t(x_1, \dots, x_M)^{s-1}}{p_{t}(x_1, \dots, x_M)^{i-1}}.$$
First, we will determine $f_1$. Notice that $T_t^{(1)}(f)=T_t(f)$ and 
$$f_1(P)=T_t(f)(P)$$ for all $P \in A_t$. Since
$f_1= \sum_{j=1}^{k} a_{1j} h_j$ for some $a_{1j} \in \F_q$ and $I \subseteq  A_1 \cup \dots \cup A_m$, we may use the values $h_j(P)$ for $j \in [k]$  and $P \in I$  to set up a system of equations
$$
\sum_{j=1}^{k} a_{1j} h_j(P) = f_1(P).
$$
Because $I$ is an information set, we can determine the $a_{1j}$ and hence 
$f_1$.

Next, induct on $i\in [m]$, assuming that $f_s$, $s \in [i-1]$, is known. Notice that $T_{t}^{(i)}(f)$ can be determined from 
$T_{t}(f)$ and $\left\{ f_s : s \in [i-1] \right\}.$ Because 
$$
T_{t}^{(i)}(f)= f_{l-m+t}(p_{t})^{l-m-i+1}+\sum_{s=i}^{l-m}f_{s}(p_{t})^{s-i}, 
$$
evaluating at $P \in A_t$ gives 
$$
f_i(P)=T_t^{(i)}(P).
$$
The rest of the argument is similar to the $s=1$ case. In particular, recall that 
$f_i= \sum_{j=1}^{k} a_{ij} h_j$ for some $a_{ij} \in \F_q$. Then using the values $h_j(P)$, $P \in I$ gives a system of $k$ equations $f_i(P)= \sum_{j=1}^{k} a_{ij} h_j(P)$ in $k$ unknowns $a_{ij}$, $j \in [k]$. Because $I$ is an information set, the $a_{ij}$ can be found, hence revealing $f_i$. In this way, $\left\{ f_s(x_1, \dots, x_M): s \in [m] \right\}$ may be determined, and $
f(x_1,\ldots,x_M) = \sum_{s=1}^l f_s(x_1,\ldots,x_M)\nu_s.$ \qed
\end{proof}

To obtain a partition of the evaluation points  for which Theorem \ref{T:recover1} applies, it is most helpful if an information set is known, especially one of cardinality $k$. In the absence of this information, one may take $I$ to be the set of all evaluation points. The downside to doing so is that then one needs $\left\lbrace T_{t}(f)(P): t \in [m], P \in I \right\rbrace$ and the full generator matrix $A$. 

The next result provides a hypothesis on the partition which is easier to check than whether a set $I$ is an information set.

\begin{corollary} \label{C:recover2}
Suppose $I:=A_{1} \dot{\cup} \cdots \dot{\cup} A_{m}\subseteq \left\{ P_1, \dots, P_n \right\}$ with $\mid I \mid \geq \deg G$ and $\sum_{P\in I}P-G$ is not principal. Then 
 $f \in \mathcal{L}_{l}(G)$ depends only on $\left\lbrace T_{t}(f)(P): t \in [m], P \in I \right\rbrace$ and $A\mid_I$.
\end{corollary}

\begin{proof}
    Consider the divisor $D'=\sum_{P\in A_1\cup\cdots\cup A_m} P$. The evaluation map restricted to the points of $D'$ has kernel $\mathcal{L}(G-D')$. If $\deg D'\geq\deg G$ and $G-D'$ is not principal, then $\mathcal{L}(G-D')=\{0\}$. Thus, $A_1\cup\cdots\cup A_m$ contains an information set. The proof follows from Theorem \ref{T:recover1}.\qed
\end{proof}

The requirement of the non-principality of $D'-G$ is unavoidable, as the following example demonstrates.

\begin{example}\rm
    Let $q=2$ and $\mathcal{X}$ be the Hermitian curve given by $x^3=y^2+y$ over $\F_{4}$. Set $G=2P_\infty$, where $P_{\infty}$ denotes the unique point at infinity on $\mathcal X$. Take $A_1\cup A_2=\{(0,0),(0,1)\}$. Notice that $\mathcal{L}(G)=\left<1, x \right>$. The evaluation of $\mathcal{L}(G)$ in the points of $A_1\cup A_2$ are 
    $$\begin{pmatrix} 0&0\\ 1&1 \end{pmatrix}$$
and thus $A_1\cup A_2$ does not form  an information set of $C((0,0)+(0,1),G)$.
\end{example}

\rmv{
\begin{example}
\glm{Add example that is not norm-trace, possibly Suzuki}

\glm{$q=8$, just include the functions $T_i(f_j)$}
\end{example}}

One may note that Theorem \ref{T:recover1} and Corollary \ref{C:recover2} recover $f \in \mathcal L_l(G)$ by utilizing an array that can be expressed as 
\begin{equation} \label{E:T_array}
\left[ 
\begin{array}{cccc} 
T_1(f)(P_1) & T_1(f)(P_2) &  \cdots & T_1(f)(P_{n'}) \\
T_2(f)(P_1) & T_2(f)(P_2) &  \cdots & T_2(f)(P_{n'}) \\
\vdots & \vdots & & \vdots 
\\
T_m(f)(P_1) & T_m(f)(P_2) &  \cdots & T_m(f)(P_{n'}) 
\end{array} \right] \in \F_q^{m \times n'}
\end{equation}
where $\left\{ P_1, \dots, P_{n'} \right\} = A_1 \cup \dots \cup A_m$. Since $m<l$, the array (\ref{E:T_array}) consists of fewer elements of $\F_q$ than the $ln$ elements required to express the coordinates of the codeword $\ev(f)$.

We will see in the next subsection how information about the annihilator functions provides a better description of the virtual projection. Then, in Section \ref{S:decoding}, we will see how these ideas support fractional decoding.

\subsection{Virtual projections from particular partitions} \label{S:vp_Kummer}

We consider virtual projections by specifying the partitions of evaluation points and associated  annihilator functions. This will allow us to better understand the virtual projection in Proposition 
\ref{P:subcode} and better control it for fractional decoding in Section \ref{S:decoding}.

Let $\mathcal{X}$ be a curve over $\mathbb{F}_q^\ell$  with a single point at infinity $P_\infty$. Let $\{P_1,\ldots,P_n\}\subseteq\mathcal{X}(\mathbb{F}_q)$ and assume there is a polynomial $z\in\mathbb{F}_q(\mathcal{X})$ such that there exists $\alpha_1,\ldots,\alpha_r\in\mathbb{F}_q$ such that
$$(z-\alpha_i)=\sum_{j=1}^{s} P_{ij}-rP_\infty$$

\noindent and $\#\{P_{ij}\ :\ i\in[r],\ j\in[s]\}=rs\leq n$. Let $A'_1\dot{\cup}\cdots\dot{\cup} A'_m$ be a partition of $\{\alpha_1,\ldots,\alpha_r\}$ and let 
$$A_t=\cup_{a\in A'_t} P_a.$$
This is equivalent to taking $A_t=\cup_{a\in A'_t}\mathrm{supp}(z-a)_0$ and we may take the annihilator polynomial of $A_t$ to be $p_t=\prod_{a\in A'_t} (z-a)$.

With this, we can understand where the virtual projection lies.

\begin{theorem}\label{T:alt}
    For any $A'_1\dot{\cup}\cdots\dot{\cup} A'_m$ as above and $t\in [m]$, the $t$-virtual projection of $\mathcal{C}_l(\beta P_\infty)$ is a subcode of $\mathcal{C}(\beta_t P_\infty)$ where $\beta_t=\beta+s \mid A'_t \mid (l-m)$. 
\end{theorem}

\begin{proof}
    Let $f\in\mathcal{L}_l(\beta P_\infty)$. Using the fact that $v_{P_\infty}(p_t)=s\mid A'_t\mid$ and $v_P(p_t)\geq 0$ for all other $P\neq P_\infty$, with Equation (\ref{E:vP}) we can see that
    $$v_P(T_t(f))\geq\begin{cases} -v_P(\beta P_\infty)& \text{if}\ v_P(p_t)\geq 0\\ -v_P(\beta P_\infty)+(l-m)v_P(p_t)&\text{if}\ v_P(p_t)<0\end{cases}.$$

    Then $v_P(T_t(f))\geq 0$ for all $P\neq P_\infty$ and $v_{P_\infty}(T_t(f))\geq -\beta-(l-m)s\mid A'_t\mid$ and we have the conclusion.
\end{proof}

As an application of Theorem \ref{T:alt}, we consider codes from a family of Kummer extensions and two natural partitions of evaluation points that give rise to virtual projections.

\begin{corollary}
Consider the curve $\mathcal X$ given by  $L(y)=x^u$ over $\F_{q^l}$ where $L(y)=\sum_{i=0}^r a_iy^{q^i} \in \F_{q^l}[y]$ is a separable, linearized polynomial with $a_0, a_r \neq 0$, and $u \mid \frac{q^l-1}{q-1}$ and the associated one-point code $\mathcal{C}_l(D, \beta P_\infty)$ where $D$ is supported by $n<\beta$ $\F_q$-rational points on $\mathcal X$. 
\begin{enumerate}
\item 
Suppose $A'_1\dot{\cup}\cdots\dot{\cup} A'_m\subseteq\mathbb{F}_{q^r}$, and let $A_i=\{P_{ab}\in\mathcal{X}(\mathbb{F}_{q^r})\ :\ a\in A'_i\}$ for $i \in [m]$. Then 
      $$
\mathcal{V}_{t}(\beta P_\infty)  \subseteq 
    C \left( (\beta + (l-m)|A'_i|q^{r-1}  \right)P_\infty).
$$    
\item
    Suppose $A'_1\dot{\cup}\cdots\dot{\cup}A'_m\subseteq\mathbb{F}_{q^r}^\ast$. Let $D$ be the sum of all $\F_{q^r}$-rational points $(a,b)$ of $\mathcal X_{q,r}$ with $a \neq 0$, and consider $A_i=\{P_{ab}\in\mathcal{X}(\mathbb{F}_{q^r})\ :\ b\in A'_i\}$, $i \in [m]$. Then the virtual projection of the code with evaluation points in the support of $D$ is
    $$
\mathcal{V}_{t}(\beta P_\infty)  \subseteq 
    C\left(\left(\beta + (l-m)|A'_i|u \right)P_\infty\right).
$$ 
\end{enumerate}   
\end{corollary}

Several well-known one-point algebraic geometry codes satisfy the hypotheses given in Theorem \ref{T:alt}, allowing them to be considered as collections of codes over $\F_q$ rather than over $\F_{q^l}$.

\begin{example}\rm
    Let $\mathbb{F}_{4}$ with $\alpha$ a primitive element such that $\alpha^2+\alpha+1=0$ and $\mathcal{X}$ be the Hermitian curve defined by $x^{3}=y^2+y$. We can define two partitions using the polynomials $x$ and $y$. 

    For any $a\in\mathbb{F}_q^\ast$ we have
    $$(x-a)_0=P_{a\alpha}+P_{a\alpha^2}$$

    \noindent and
    $$(x)_0=P_{00}+P_{01}.$$

    Thus, we can use any partition of $\mathbb{F}_q$ to define a virtual projection as in Theorem \ref{T:alt} for any code $\mathcal{C}(\beta P_\infty)$. 

    We can do the same with a partition of $\{\alpha,\alpha+1\}$ since for $a=\alpha,\alpha+1$, we have $(y-a)_0=P_{1a}+P_{\alpha a}+P_{\alpha+1 a}$. We can use this to make a virtual projection satisfying the conditions of Theorem \ref{T:recover1} for any $G=\beta P_\infty$ with $\beta<8$.

\end{example}

\section{Decoding via virtual projections} \label{S:decoding}

In this section, we detail how virtual projections give rise to fractional decoding algorithms for algebraic geometry codes over extension fields. First, we consider fractional decoding directly from virtual projections in Subsection \ref{s:directly}. Next, in Subsection \ref{S:interleaved_vp}, we use virtual projections to define related interleaved codes, which gives rise to another approach to fractional decoding in Subsection \ref{s:decoding_interleaved}.

Throughout, we consider a curve $\mathcal X$ over $\F_{q^l}$ with 
divisors $G$ and $D$ whose supports consist only of $\F_q$-rational points and the algebraic geometry code  $\mathcal{C}_{l}(G,D) \subseteq \F_{q^l}^n$. As before, we also take $\left\{ \zeta_{1},\ldots,\zeta_{l} \right\}$ to be a basis of the extension $\F_{q^l}/\F_q$ with $\left\{ \nu_{1},\ldots,\nu_{l} \right\}$ its dual basis.

\subsection{Fractional decoding directly from virtual projections}
\label{s:directly}

Given a received word $w=\ev(f)+e \in \F_{q^{ln}}$ where $f \in \mathcal L_l(G)$, for each $i \in [n]$ and $t \in [m]$, download
\begin{equation} \label{E:dit}
w_i^t:=tr \left( \zeta_{l-m+t} w_i \right)p_t(P_i)^{l-m} + \sum_{s=1}^{l-m} tr \left( \zeta_{s} w_i \right)p_t(P_i)^{s-1}; 
\end{equation}
We aim to show that $\ev(f)$ can be recovered from 
\begin{equation} \label{E:dit_array}
\pi(w):=\left[ 
\begin{array}{cccc} 
w_1^1 & w_2^1 &  \cdots & w_{n'}^1 \\
w_1^2 & w_2^2 &  \cdots & w_{n'}^2 \\
\vdots & \vdots & & \vdots 
\\
w_1^m & w_2^m &  \cdots & w_{n'}^m
\end{array} \right] \in \F_q^{m \times n'}
\end{equation}
if $wt(e)$ is not too large.
Notice that if $w=\ev(f)$, meaning $wt(e)=0$, then 
$$
\begin{array}{lcl}
w_i^t&=&tr \left( \zeta_{l-m+t} f(P_i) \right)p_t(P_i)^{l-m} + \sum_{s=1}^{l-m} tr \left( \zeta_{s} f(P_i) \right)p_t(P_i)^{s-1} \\ \ \\
&=& f_{l-m+t} (P_i) p_t(P_i)^{l-m} + \sum_{s=1}^{l-m}f_{s}(P_i)p_t(P_i)^{s-1} \\ \ \\
& = & T_t(f)(P_i).
\end{array}
$$
Indeed,
if $f(x_1,\ldots,x_M)=\sum_{j=1}^k a_j h_j(x_1,\ldots,x_M)$ where $h_j \in \F_{q}[x]$, then 
$$tr \left( \zeta_s f(x_1,\ldots,x_M) \right) = tr \left( \zeta_s \sum_{j=1}^k a_j h_j(x_1,\ldots,x_M) \right)=tr \left( \zeta_s \sum_{j=1}^k a_j \right)h_j(x_1,\ldots,x_M).$$
We have already seen in Theorem \ref{T:recover1} that $\ev(f)$ can be recovered from the array 
$$
\left[ 
\begin{array}{cccc} 
T_1(f)(P_1) & T_1(f)(P_2) &  \cdots & T_1(f)(P_{n'}) \\
T_2(f)(P_1) & T_2(f)(P_2) &  \cdots & T_2(f)(P_{n'}) \\
\vdots & \vdots & & \vdots 
\\
T_m(f)(P_1) & T_m(f)(P_2) &  \cdots & T_m(f)(P_{n'}) 
\end{array} \right] \in \F_q^{m \times n'}.
$$
The challenge now is to confirm that this is the case if $w=\ev(f)+e$ where $e \in \F_{q^l}$ has $0<wt(e)<B$, meaning some (but fewer than $B$) positions are in error, for some specified value $B$. Given that $mn<ln$, it is reasonable to suspect that $B<\frac{n-\deg G}{2}$, the number of correctable errors for the code. This challenge is addressed in the next result where we see the value of $B$ depends on the degree of the pole divisor of an annihilator function.

\begin{theorem} \label{T:main_decoding}
Suppose $I \subseteq A_{1} \dot{\cup} \cdots \dot{\cup} A_{m}\subseteq \left\{ P_1, \dots, P_n \right\}$,
  for some $m \in \N$ and some information set $I$ for the code $C_l(G)$.
    The code $C_l(D,G)$ can correct any 
\begin{equation}
 \label{e_value} \left\lfloor 
\frac{  n-\left( \deg G + (l-m) \max \left\{ \deg (p_t)_{\infty}: t \in [m] \right\} \right)}{2} \right \rfloor
\end{equation}
errors using $mn'$ entries of $\F_q$.
\end{theorem}

\begin{proof} 
Consider a received word $w$ in which at most $d'$ errors have occurred, where $d'$ denotes the value in (\ref{e_value}), 
is the meaning $w=ev(f)+e$ where $wt(e) \leq d'$. Let $\mathcal E = \left\{ i \in [n]: e_i \neq 0 \right\}$ denote the set of positions in error. 

Download the $w_i^t$, $i \in [n']$, $t \in [m]$ as in Equation (\ref{E:dit_array}).
Notice that $w_i^t=T_t(f)(P_i)+e_{ti}$ for some $e_{ti} \in \F_q$. Moreover, $e_{ti}\neq 0$ only if $i \in \mathcal E$. 
Thus, we may consider $$Row_t \left( \pi(w) \right) = ev \left( T_t(f) \right) + \left( e_{t1}, \dots, e_{tn'} \right)$$ as a received word in $C(G_t)\subseteq \F_q^{n'}$ where 
$$
G_t=G + (l-m) \left(p_t\right)_{\infty}.
$$
Since the minimum distance of $C(G_t)$ is at least $n-\deg G_t$, at least $\frac{n-\deg G_t}{2}$ errors may be corrected in each row $Row_t \left( \pi(w) \right)$, giving rise to  $ev \left( T_t(f) \right)$
if there are no more than $\frac{n-\deg G_t}{2}$ errors. Hence, the up to $d'$ errors in each  
$Row_t \left( \pi(w) \right)$ are corrected to give $ev \left( T_t(f) \right)$. This gives an array as shown in (\ref{E:T_array}). 

Now, we apply Theorem \ref{T:recover1} to this  array to obtain $\ev(f) \in \F_{q^l}^n$.
\end{proof}

Observe that this fractional decoding consists of the following steps for a received word $w \in \F_{q^l}^n$:
\begin{enumerate}
    \item Download $\pi(w) \in \F_q^{m \times n'}$ as in Equation (\ref{E:dit_array}).
    \item Apply a decoding algorithm for $C(G+G_t)\subseteq \F_q^{n'}$ to $w^t$ to obtain $T_t(f)$ for all $t \in [m]$.
    \item Apply Theorem \ref{T:recover1} to obtain $\ev(f) \in \F_{q^l}^{n'}$.
\end{enumerate}
The performance will be dictated by the partition  of the evaluation set $A_{1} \dot{\cup} \cdots \dot{\cup} A_{m}\subseteq \left\{ P_1, \dots, P_n \right\}$ and associated choice of annihilator functions. Moreover, one may select among the various decoding algorithms for algebraic geometry codes for Step (1) and evolve with related advances such as \cite{Puchinger_power_decoding}.

While Theorem \ref{T:main_decoding} is stated in terms of an  information set, the hypothesis of Corollary \ref{C:recover2} allows us to give the next, similar result which does not rely on knowledge of an information set. 

\begin{corollary}
Suppose $I:=A_{1} \dot{\cup} \cdots \dot{\cup} A_{m}\subseteq \left\{ P_1, \dots, P_n \right\}$ with $\mid I \mid \geq \deg G$ and $\sum_{P\in I}P-G$ is not principal. Then the code $C_l(D,G)$ can correct any 
$$
\left\lfloor 
\frac{  n-\left( \deg G + (l-m) \max \left\{ \deg (p_t)_{\infty}: t \in [m] \right\} \right)}{2} \right \rfloor
$$
errors using $mn$ entries of $\F_q$.
\end{corollary}

\begin{proof}
    According to the proof of Corollary \ref{C:recover2}, $A_{1} \dot{\cup} \cdots \dot{\cup} A_{m}$ contains an information set. The result now follows as in the proof of Theorem \ref{T:main_decoding}.
\end{proof}

\subsection{Interleaved codes from virtual projections} \label{S:interleaved_vp}

The use of interleaved codes in fractional decoding can, with a certain probability, provide a larger decoding radius than non-interleaved codes. In this subsection, we recast the virtual projections of algebraic geometry codes as interleaved codes. We begin with recalling useful terminology.

Given a collection of codes $C_1, \dots, C_m \subseteq \F_q^n$, one may define an interleaved code 
$$
\mathcal{IC}(C_1, \dots, C_m)
:=
\left\lbrace \left[\begin{array}{c}
c_1 \\
c_2 \\
\vdots \\
c_m
\end{array}\right]: c_i \in C_i, i \in [m]
\right\rbrace.$$
If $C_i=C$ for all $i\in [m]$, then $\mathcal{IC}(C_1, \dots, C_m)$ is said to be homogeneous; otherwise, $\mathcal{IC}(C_1, \dots, C_m)$ is called heterogeneous.

Given a curve $X$ over a finite field $\F$, an interleaved algebraic geometry code on $\mathcal X$ is of the form 
$$
\mathcal{IC}(D, G_1, \dots, G_m)
:=
\left\lbrace \left[\begin{array}{c}
c_1 \\
c_2 \\
\vdots \\
c_m
\end{array}\right]: c_i \in C(D,G_i), i \in [m]
\right\rbrace.$$
If $G_i=G$ for all $i \in [m]$, then we write $\mathcal{IC}(D,G,m)$ for this homogeneous code.

According to Proposition \ref{P:subcode}, we may consider the set of arrays in (\ref{E:T_array}) as  subcodes of interleaved codes. In particular, it may be considered 
 as a subcode of the heterogeneous interleaved code $\mathcal{IC}(D, G_1, \dots, G_m)$ where $$G_t:= G + (l-m)  \left(p_t(x_1,\ldots,x_M) \right)_{\infty} P$$
 for all $t \in [m]$: 
$$
\left\{ 
\left[ 
\begin{array}{c}
\ev(T_1(f))\\
\ev(T_2(f))\\
\vdots\\
\ev(T_m(f)) 
\end{array} \right] : f \in \mathcal L_l(G) \right\}
\subseteq \mathcal{IC}(D, G_1, \dots, G_m).$$ Alternatively, it may be viewed as a subcode of the homogeneous interleaved code $\mathcal{IC}(D, G, m)$
where $$G=\sum_{P \in S} \left( v_P(G)-(l-m)v_P(p_t) \right)$$ and $S=\cup_{t \in [m]} \supp (p_t)_{\infty}$.
: $$
\left\{ 
\left[ 
\begin{array}{c}
\ev(T_1(f))\\
\ev(T_2(f))\\
\vdots\\
\ev(T_m(f)) 
\end{array} \right] : f \in \mathcal L_l(G) \right\}
\subseteq \mathcal{IC}(D, G, m).$$

We will now see that viewing the collection of virtual projections of certain algebraic geometry codes as interleaved codes facilitates fractional decoding.

\subsection{Fractional decoding via interleaved algebraic geometry codes} \label{s:decoding_interleaved}

In this subsection, we provide fractional decoding algorithms for certain one-point algebraic geometry codes from constant field extension by viewing their virtual projections as interleaved codes.

In \cite{Brown}, it is shown that interleaved one-point algebraic geometry codes $C(G)$ of length $n$ constructed from a curve of genus $g$ can be decoded up to $e$ errors where  $e \leq (n-\deg G-g)b-\frac{c}{al+1}$ for some $c>0$, $a:=\frac{\ln{(q^l-1)}}{\ln{q^l}}$, and $b:=\frac{la}{la+1}$.  The same decoder may be used in the case of $e$ errors where $e \leq g-1$ and succeeds with probability at least $1-\frac{q}{q^c(q-1)}$. We will see how this decoder, and those tailored to codes from particular codes, aid in fractional decoding. 

Consider the code $C_l(\beta P)$ where 
$$T_t \left( \mathcal L \left(\beta P \right) \right) \subseteq \mathcal L \left( \alpha P\right)\  \forall t \in [m].$$
Fix a parameter $t' \in \N$. Then $\mathcal L \left( \left(g+t' \right) P \right) \subseteq \mathcal L \left( \left(g+t'+\alpha \right) P \right)$. Take a basis $\left\{ \varphi_1, \dots, \varphi_s \right\}$ for $\mathcal L \left( \left(g+t'+\alpha \right) P \right)$ so that $\left\{ \varphi_1, \dots, \varphi_{s'} \right\}$ is a basis for $\mathcal L \left( \left(g+t' \right) P \right)$. Then the transposes of 
$$
V:=\left[ 
\begin{array}{ccc}
\varphi_1(P_1) & \cdots & \varphi_s(P_1) \\
\varphi_1(P_2) & \cdots & \varphi_s(P_2) \\
\vdots & & \vdots \\
\varphi_1(P_n) & \cdots & \varphi_s(P_n) 
\end{array}
\right]
\mathrm{\ and \ } 
W:=\left[ 
\begin{array}{ccc}
\varphi_1(P_1) & \cdots & \varphi_{s'}(P_1) \\
\varphi_1(P_2) & \cdots & \varphi_{s'}(P_2) \\
\vdots & & \vdots \\
\varphi_1(P_n) & \cdots & \varphi_{s'}(P_n) 
\end{array}
\right]
$$
are generator matrices of  $C \left( \left(g+t'+\alpha \right) P \right)$ and $C\left( \left(g+t'\right) P \right)$
Determine a partition as in Theorem \ref{T:recover1} along with the associated annihilator polynomials $p_t$, $t \in [m]$. 

Now suppose $w=ev(f)+e$ is received for some $f \in \mathcal L_l(\beta P)$. 
From the projection of $w$ given in Equation \ref{E:dit_array}, for each $i \in [m]$, form the matrix  
\begin{equation} \label{E:W_array}
W_i:=- \left[ 
\begin{array}{cccc} 
w_1^i \varphi_1(P_1) & w_1^i \varphi_2(P_1) &  \cdots & w_1^i \varphi_{s'}(P_1) \\
w_2^i \varphi_1(P_1) & w_2^i \varphi_2(P_1) &  \cdots & w_2^i \varphi_{s'}(P_1)\\
\vdots & \vdots & & \vdots \\
w_n^i \varphi_1(P_1) & w_n^i \varphi_2(P_1) &  \cdots & w_n^i \varphi_{s'}(P_1)
\end{array} \right] \in \F_q^{m \times n'}.
\end{equation} and for $i \in [l] \setminus [m]$, let $W_i \in \F_q^{n \times n}$ be the zero matrix. 
Set $$A:=\left[ 
\begin{array}{cccc|c}
V & & & & W_1 \\
& V & & & W_2 \\
& & \ddots & & \vdots \\
& & & V & W_l
\end{array}
\right].$$
Given $v=(v_1; \dots; v_l; v_{l+1}) \in NS(A) \cap \F_q^{ls+s'}$, set $$T^t(f):=\frac{\sum_{j=1}^s v_{tj} \varphi_j}{\sum_{j=1}^{s'} v_{tj} \varphi_j}$$ for all $t \in [m]$.
Then apply Theorem \ref{T:recover1} with the $T^t(f)$ in place of the $T_t(f)$ to recover $ev(f).$
According to \cite[Theorem 1]{Brown}, this proves the following. 

\begin{proposition} \label{P:interleaved}
Consider the code $C_l(\beta P)$ where 
$T_t \left( \mathcal L \left(\beta P \right) \right) \subseteq \mathcal L \left( \alpha P\right)\  \forall t \in [m]$.
   A received word $w=\ev(f)+e \in \F_{q^{ln}}$, where $f \in \mathcal L_l(\alpha P)$ for some $\F_q$-rational point $P$, 
can be decoded from $mn<ln$ symbols of $\F_q$ provided $wt(e) \leq (n-\alpha-g)b-\frac{c}{al+1}$ for some $c>0$. If 
 $(n-\alpha-g)b-\frac{c}{al+1} < e \leq g-1$, then decoder failure is possible but the decoder succeeds with probability at least $1-\frac{q}{q^c(q-1)}$.
\end{proposition}

While Proposition \ref{P:interleaved} applies to a large family of algebraic geometry codes, advances have been made for codes from particular curves.

Recall that the algorithm downloads symbols which are $m$ corrupted codewords of $\mathcal{C}(\beta_{r}P_{\infty})$, where $\beta_{r}=\beta+q\left\lceil\frac{\beta+1}{q}\right\rceil\left(\frac{1-\alpha}{\alpha}\right)$. Those $m$ corrupted codewords can be arranged to form the following matrix

\begin{equation}
\pi(h)=\left[\begin{array}{llll}
T_{0}(h)(P_1)&T_{0}(h)(P_2)&\ldots&T_{0}(h)(P_n)\\
T_{1}(h)(P_1)&T_{1}(h)(P_2)&\ldots&T_{1}(h)(P_n)\\
\vdots&\vdots&\ddots&\vdots\\
T_{m-1}(h)(P_1)&T_{m-1}(h)(P_2)&\ldots&T_{m-1}(h)(P_n)
\end{array}\right]\in M_{m,n}(\F_{q^r}).
\end{equation}

The matrix $\pi(h)$ is called the homogeneous virtual projection of $h$. Moreover, the downloaded symbols may be viewed as corrupted a codeword of a homogeneous interleaved norm-trace code

\begin{equation}
\mathcal{C}_{\mathcal{X}}(n,m,\beta_{r})=\left\lbrace \left[\begin{array}{c}
c^{(0)}\\
\vdots\\
c^{(m-1)}
\end{array}\right]\in M_{m,n}(\F_{q^r}):\begin{array}{l}c^{(i)}\in\mathcal{C}(\beta_{r}P_{\infty})\\ \beta_{r}=\beta+q\left\lceil\frac{\beta+1}{q}\right\rceil\left(\frac{1-\alpha}{\alpha}\right)\end{array}
\right\rbrace.
\end{equation}

So, we point out that it is possible to apply the power decoding algorithm present in \cite{Puchinger_Rosenkilde_Bouw} to improve the $\alpha$-decoding correction capability (with some failure probability) of our decoding algorithm for certain families of codes.

In the same way, when we are considering a small constant extension it is possible to use the downloaded symbols  to create the following matrix
\begin{equation}
\pi^{\star}(h)=\left[\begin{array}{llll}
R_{0}(h)(P_1)&R_{0}(h)(P_2)&\ldots&R_{0}(h)(P_n)\\
R_{1}(h)(P_1)&R_{1}(h)(P_2)&\ldots&R_{1}(h)(P_n)\\
\vdots&\vdots&\ddots&\vdots\\
R_{m-1}(h)(P_1)&R_{m-1}(h)(P_2)&\ldots&R_{m-1}(h)(P_n)
\end{array}\right]\in M_{m,n}(\F_{q^r}).
\end{equation}

The matrix $\pi(h)$ is called the homogeneous virtual projection of $h$. Moreover, the downloaded symbols may be viewed as corrupted codeword of a homogeneous interleaved norm-trace code

\begin{equation}
\mathcal{C}_{\mathcal{X}}(n,m,\beta^{\star})=\left\lbrace \left[\begin{array}{c}
c^{(0)}\\
\vdots\\
c^{(m-1)}
\end{array}\right]\in M_{m,n}(\F_{q^r}):\begin{array}{l}c^{(i)}\in\mathcal{C}(\beta_{r}P_{\infty})\\ \beta^{\star}=\beta+q(q+1)(l-m)
\end{array}
\right\rbrace.
\end{equation}
So, we can apply the power decoding algorithm present in \cite{Puchinger_Rosenkilde_Bouw} to improve the $\alpha$-decoding correction capability (with some failure probability) of our decoding algorithm.

Analogous to fractional decoding directly from virtual projections and advances in decoding algebraic geometry codes, fractional decoding of algebraic geometry codes via projections represented as interleaved codes can make use of emerging advances in the decoding of interleaved algebraic geometry codes.


\section*{Acknowledgments}
{The first and second author were partially supported by NSF DMS-2201075 and the Commonwealth Cyber Initiative. The second and third author were partially supported by DMS-1802345.}

\bibliographystyle{plain}
\bibliography{FDONTC}

\end{document}